\documentclass[12pt]{article}
\usepackage{authblk}
\usepackage{color}
\usepackage[all]{xy}
\usepackage{subcaption}
\usepackage{amsmath,epsfig}
\usepackage{amssymb,amsthm}
\usepackage{setspace}
\usepackage{mathrsfs}
\usepackage{lineno}

\numberwithin{equation}{section}
\newcommand{\RR}{\rm I\kern -1.6pt{\rm R}}

\newtheorem{theorem}{Theorem}[section]
\newtheorem{lemma}[theorem]{Lemma}

\newtheorem{corollary}[theorem]{Corollary}

\newcommand{\diag}{\mathop{\mathrm{diag}}}
\newcommand{\R}{\mathcal{R}}

\title{\bf Fast Diffusion Inhibits Disease Outbreaks }
 
\author{{\bf Daozhou Gao$^{a}$\thanks{Corresponding author. E-mail address: dzgao@shnu.edu.cn}, Chao-Ping Dong$^{a}$}\\
{\small $^{a}$ Mathematics and Science College, Shanghai Normal University, Shanghai, 200234 China}
}

\date{}

\begin{document}
\maketitle

\noindent  {\bf Abstract.}
We show that the basic reproduction number of an SIS patch model with standard incidence is either strictly decreasing and strictly convex with respect to the diffusion coefficient of infected subpopulation if the patch reproduction numbers of at least two patches in isolation are distinct or constant otherwise. Biologically, it means that fast diffusion of the infected people reduces the risk of infection. This completely solves and generalizes a conjecture by Allen et al. ({\it SIAM J Appl Math}, 67: 1283-1309, 2007). Furthermore, a substantially improved lower bound on the multipatch reproduction number, a generalized monotone result on the spectral bound the Jacobian matrix of the model system at the disease-free equilibrium, and the limiting endemic equilibrium are obtained. The approach and results can be applied to a class of epidemic patch models where only one class of infected compartments migrate between patches and one transmission route is involved. \\

\noindent {\bf AMS subject classifications.} 91D25, 34D20, 92D30, 34D05, 15B48, 15A42.\\

\noindent {\bf Key words.} patch model, basic reproduction number, monotonicity, diffusion coefficient, spectral bound, essentially nonnegative matrix.

\section{Introduction}

In 2007, Allen and her collaborators proposed the following SIS epidemic patch model
\begin{equation}\label{sismodel}
\begin{aligned}
\dfrac{dS_i}{dt} &= d_S\sum\limits_{j\in\Omega} L_{ij} S_j-\beta_i\dfrac{S_iI_i}{S_i+I_i}+\gamma_i I_i,\ i\in\Omega,\\
\dfrac{dI_i}{dt} &= d_I\sum\limits_{j\in\Omega} L_{ij} I_j+\beta_i\dfrac{S_iI_i}{S_i+I_i}-\gamma_i I_i,\ i\in\Omega,
\end{aligned}
\end{equation}
where $\Omega=\{1,2,\dots,n\}$ and $n\ge 2$ is the number of patches. The variables $S_i(t)$ and $I_i(t)$ represent the number of susceptible and infected individuals in patch $i$ at time $t$, respectively. The parameters $\beta_i$ and $\gamma_i$ are positive transmission coefficient and recovery rate in patch $i$, respectively; $d_S$ and $d_I$ are positive diffusion coefficients for the susceptible and infected subpopulations, respectively; $L_{ij}$ is a nonnegative constant that denotes the degree of movement from patch $j$ to patch $i$ for $i\neq j$ and $-L_{ii}=\sum^n_{j=1,j\neq i}L_{ji}$ is the degree of movement from patch $i$ to all other patches.

The following three assumptions on the initial condition, the connectivity matrix $L=(L_{ij})$, and the patch reproduction number $\R_0^{(i)}=\beta_i/\gamma_i$ are made:
\begin{itemize}
  \item [(A1)] $S_i(0)\ge0$ and $I_i(0)\ge0$ for $i\in\Omega$, and $\sum_{i\in\Omega} I_i(0)>0$;
  \item [(A2)] $L$ is essentially nonnegative (or called quasi-positive), irreducible, and symmetric;
  \item [(A3)] $H^-=\{i\in\Omega:\R_0^{(i)}<1\}$ and $H^+=\{i\in\Omega:\R_0^{(i)}>1\}$ are nonempty and $H^-\cup H^+=\Omega$.
\end{itemize}
It follows from Theorem 6.4.16 in Berman and Plemmons \cite{berman1994nonnegative} that $L$ has rank $n-1$ and hence the system of linear equations
\[\sum_{j\in\Omega} L_{ij}S_j=0,\ i=1,\dots,n, \mbox{ and } \ \sum_{i\in\Omega} S_i=\sum_{i\in\Omega}(S_i(0)+I_i(0))\]
has a unique positive solution, denoted by $\boldsymbol{S}^0$. Then the model \eqref{sismodel} admits a unique disease-free equilibrium (DFE) $E_0=(\boldsymbol{S}^0,\boldsymbol{0})$. Linearizing the model system \eqref{sismodel} at the DFE gives the new infection and transition matrices
\[F=\diag\{\beta_1,\dots,\beta_n\}\ \mbox{ and }\ V=D-d_IL=\diag\{\gamma_1,\dots,\gamma_n\}-d_IL,\]
where $D=\diag\{\gamma_1,\dots,\gamma_n\}$. Following the recipe of van den Driessche and Watmough \cite{vandendriessche-watmough2002}, the basic reproduction number for model \eqref{sismodel} is defined as the spectral radius of the next generation matrix (Diekmann et al. \cite{diekmann1990definition}) $FV^{-1}$, i.e.,
\[\R_0=\rho(FV^{-1}).\]

Allen et al. \cite{allen2007asymptotic} showed that the DFE is globally asymptotically stable if $\R_0<1$ and there exists a unique endemic equilibrium if $\R_0>1$. Under assumptions (A1)-(A3), two main theorems linked spatial heterogeneity, habitat connectivity and movement rate to disease dynamics are presented. Three open problems are left in their discussion. The first one is to conjecture that the basic reproduction number $\R_0$ is a monotone decreasing function of $d_I$. Biologically speaking, an increase in the diffusion of infected subpopulation can lower the potential for disease transmission. The two-patch case can be easily verified by direct calculation. Nevertheless, when three or more patches are concerned, the expression of $V^{-1}$ is complicated so that a direct proof of the monotonicity is intractable. Recently, Gao \cite{gao2019travel} gave an affirmative answer to the conjecture by using the Perron-Frobenius theorem. The proof strongly relies on the symmetry of connectivity matrix $L$. The main purpose of the present paper is to extend the conjecture to asymmetric $L$ and to seek its application.

The remainder of this paper is organized as follows. In Section 2, based on some profound results on the spectral theory of nonnegative matrices, the basic reproduction number $\R_0$ is shown to be strictly decreasing and strictly convex in $d_I$ even if the connectivity matrix $L$ is asymmetric. Section 3 is devoted to the application of the monotonicity to estimate $\R_0$ and spectral bound of $F-V$. A brief discussion is given at the end.

\section{Monotonicity of $\R_0$}

Throughout this paper, unless otherwise indicated, we assume that:
\begin{itemize}
  \item [{\rm(B1)}] the connectivity matrix $L$ is essentially nonnegative and irreducible;
  \item [{\rm(B2)}] at least two patch reproduction numbers are different, i.e., there exist $i\neq j$ such that $\R_0^{(i)}\neq\R_0^{(j)}$ (otherwise, by Proposition 2.2 in Gao and Ruan \cite{gao2011sis}, the multipatch reproduction number $\R_0$ is constant irrespective of $L$ and $d_I$).
\end{itemize}

Now we provide a simpler proof for the conjecture of Allen et al. \cite{allen2007asymptotic} than that of Gao \cite{gao2019travel}. The single and double prime symbols denote the first and second derivatives, respectively.

\begin{theorem}
For model \eqref{sismodel}, if the connectivity matrix $L$ is symmetric, then the basic reproduction number $\R_0$ is strictly decreasing in $d_I\in[0,\infty)$ and $\R_0'(d_I)<0$ for $d_I\in(0,\infty)$.
\end{theorem}

\begin{proof}
The fact $\R_0=\rho(FV^{-1})=\rho(V^{-1}F)$ implies that there exists a column vector $\boldsymbol{v}:=\boldsymbol{v}(d_I)=(v_1,\dots,v_n)^T\gg \boldsymbol{0}$ such that $V^{-1}F\boldsymbol{v}=\R_0\boldsymbol{v}$, or equivalently,
\begin{equation}\label{eqn1}
\begin{aligned}
\left(\frac{1}{\R_0}F-V\right)\boldsymbol{v}=
\left(\frac{1}{\R_0}F-D+d_IL\right)\boldsymbol{v}=0.
\end{aligned}
\end{equation}
Differentiating both sides of \eqref{eqn1} with respect to $d_I$ gives
\begin{equation}\label{eqn2}
\begin{aligned}
\left(-\frac{\R_0'}{\R_0^2}F+L\right)\boldsymbol{v}
+\left(\frac{1}{\R_0}F-D+d_IL\right)\boldsymbol{v}'=0.
\end{aligned}
\end{equation}
Multiplying \eqref{eqn1} by $(\boldsymbol{v}')^T$ and \eqref{eqn2} by $\boldsymbol{v}^T$, and subtracting the two resulting equations yield
\[\boldsymbol{v}^T\left(-\frac{\R_0'}{\R_0^2}F+L\right)\boldsymbol{v}=0\]
due to the symmetry of $\frac{1}{\R_0}F-D+d_IL$. We thus have
\[\R_0'=\dfrac{\boldsymbol{v}^T L \boldsymbol{v}}{\boldsymbol{v}^T F \boldsymbol{v}}\R_0^2.\]
It follows from the symmetry of $L$ that
\begin{equation*}
\begin{aligned}
&\boldsymbol{v}^T L \boldsymbol{v}=\sum\limits^n_{i=1}\sum\limits^n_{j=1}L_{ij}v_iv_j
=\sum\limits^n_{i=1}\sum\limits_{j\neq i}L_{ij}v_iv_j+\sum\limits^n_{i=1}L_{ii}v_i^2 \\
=& \sum\limits^n_{i=1}\sum\limits_{j\neq i}L_{ij}v_iv_j-\sum\limits^n_{i=1}\sum\limits_{j\neq i}L_{ji}v_i^2=
\sum\limits^n_{i=1}\sum\limits_{j\neq i}L_{ij}v_i(v_j-v_i)\\
=&\sum\limits^n_{i=1}\sum\limits_{j\neq i}L_{ij}v_j(v_i-v_j)=
-\frac{1}{2}\sum\limits^n_{i=1}\sum\limits_{j\neq i}L_{ij}(v_i-v_j)^2\le0.
\end{aligned}
\end{equation*}
Similar to the proof of Lemma 3.4 in Allen et al. \cite{allen2007asymptotic}, we can
use the irreducibility of $L$ to prove by contradiction that $\boldsymbol{v}^T L \boldsymbol{v}<0$. In particular, if $v_1=\dots=v_n$ then \eqref{eqn1} implies that $\R_0^{(i)}=\R_0$ for $1\le i\le n$, a contradiction. Hence $\R_0'(d_I)<0$ for $d_I\in(0,\infty)$.
\end{proof}

Before stating the general result on the strict monotonicity of $\R_0$ with respect to $d_I$ in case of asymmetric $L$, we introduce a lemma on the spectral bound of a class of essentially nonnegative matrices.

\begin{lemma}[Theorem 1 in Altenberg \cite{altenberg2012resolvent}, Theorem 1.1 in Altenberg \cite{altenberg2013ordering}, Theorem 5.2 in Karlin \cite{karlin1982classifications}]\label{lem23}
Let $P$ be an irreducible stochastic matrix (i.e., nonnegative and each column summing to one), and let $D$ be a positive diagonal matrix that is not a scalar multiple of identity matrix $\mathbb{I}_n$ of order $n\ge 2$. Put
\[ M(\alpha)=(1-\alpha)\mathbb{I}_n+\alpha P. \]
Then for $\alpha>0$, the spectral bound $s(M(\alpha)D)$ has the following properties:
\begin{itemize}
\item[{\rm(a)}] $\frac{d}{d\alpha}s(M(\alpha)D)<0$. Thus $s(M(\alpha)D)$  decreases strictly as $\alpha$ increases.
\item[{\rm(b)}] $s(M(\alpha)D)$ is strictly convex in $\alpha$. Thus $\frac{d^2}{d\alpha^2}s(M(\alpha)D)\ge0$.
\end{itemize}
\end{lemma}

\begin{proof}
By the implicit function theorem, $s(M(\alpha)D)$, the spectral bound of the essentially nonnegative matrix $M(\alpha)D$, is twice differentiable with respect to $\alpha\in(0,\infty)$. Part (a) comes from the proof of Theorem 2 of Altenberg \cite{altenberg2009evolutionary}, which uses the results of Friedland and Karlin \cite{friedland1975some}, Friedland \cite{friedland1981convex}, and Karlin \cite{karlin1982classifications}.

Part (b) comes from the proof of Karlin's Theorem 5.2 by Altenberg \cite{altenberg2012resolvent}. For the convenience of readers, let us outline the argument. Note that
\[ M(\alpha) D=(\alpha (P-\mathbb{I}_n)+\mathbb{I}_n)D=\alpha(P-\mathbb{I}_n)D+\beta D=\alpha A+\beta D,\]
where $A:= (P-\mathbb{I}_n)D$ is an essentially nonnegative matrix and $\beta=1$. Now let $\beta$ vary in the interval $[0,+\infty)$. By Theorem 4.1 of Friedland \cite{friedland1981convex} (which strengthens the work of Cohen \cite{cohen1981convexity}), the spectral bound $s(\alpha A+\beta D)$ is strictly convex in $D$ and hence in $\beta$ as well. Then by  Lemma 1 on dual convexity in Altenberg \cite{altenberg2012resolvent}, we have that $s(\alpha A+\beta D)$ is strictly convex in $\alpha$, which also implies that $s(\alpha A+\beta D)$ is strictly decreasing in $\alpha$.
\end{proof}

Next we remove the restriction on the symmetry of the connectivity matrix $L$. The basic reproduction number $\R_0$ for model \eqref{sismodel} is found to be not only strictly decreasing but also strictly convex in $d_I\in[0,\infty)$.

\begin{theorem}\label{R0monotonicity}
For model \eqref{sismodel}, the basic reproduction number $\R_0$ is strictly decreasing and strictly convex in $d_I\in[0,\infty)$. Moreover, $\R_0'(d_I)<0$ and $\R_0''(d_I)>0$ for $d_I\in(0,\infty)$.
\end{theorem}

\begin{proof}
Denote $\tilde{D}=DF^{-1}$ and $\tilde{L}=LF^{-1}$ and $\tilde{V}=\tilde{D}-d_I\tilde{L}$. By the Perron-Frobenius theorem \cite{horn2013matrix}, there is a unique real vector $\boldsymbol{v}\gg0$ such that
\[FV^{-1}\boldsymbol{v}=(DF^{-1}-d_ILF^{-1})^{-1}\boldsymbol{v}=\tilde{V}^{-1}\boldsymbol{v}=\R_0\boldsymbol{v},\]
which implies that
\begin{equation*}\label{eqn21}
\begin{aligned}
\dfrac{1}{\R_0}\boldsymbol{v}=\tilde{V}\boldsymbol{v},
\end{aligned}
\end{equation*}
or equivalently,
\[(k\mathbb{I}_n-\tilde{V})\boldsymbol{v}=\left(k-\dfrac{1}{\R_0}\right)\boldsymbol{v}\ \mbox{ for } k\in\mathbb{R}.\]
Clearly, the square matrix
\[k\mathbb{I}_n-\tilde{V}=(k\mathbb{I}_n-\tilde{D})+d_I\tilde{L}
=(k\mathbb{I}_n-DF^{-1})+d_ILF^{-1}\]
is nonnegative and irreducible for sufficiently large $k$. Thus
\[\rho(d_I):=\rho(k\mathbb{I}_n-\tilde{D}+d_I\tilde{L})=k-\dfrac{1}{\R_0(d_I)},\]
or equivalently,
\[\R_0(d_I)=\dfrac{1}{k-\rho(d_I)}. \]
Therefore, the first and second derivatives of $\R_0$ with respect to $d_I$ are respectively
\begin{equation}\label{dR0}
\begin{aligned}
\R_0'(d_I)=\dfrac{\rho'(d_I)}{(k-\rho(d_I))^2}
\end{aligned}
\end{equation}
and
\begin{equation}\label{ddR0}
\begin{aligned}
\R_0''(d_I)=\frac{(k-\rho(d_I))\rho''(d_I)+2 (\rho'(d_I))^2}{(k-\rho(d_I))^3}.
\end{aligned}
\end{equation}

Choose $k$ large enough so that all the diagonal entries of
\[ \hat{D}:= k\mathbb{I}_n-\tilde{D} \]
are positive and the matrix
\[ \hat{P}:= \mathbb{I}_n+\tilde{L}\hat{D}^{-1} \]
is irreducible and stochastic. Note that $\hat{D}$ is not a scalar multiple of the identity matrix $\mathbb{I}_n$.  By Lemma \ref{lem23}, the spectral radius
\[\rho(d_I)=\rho(k\mathbb{I}_n-\tilde{D}+d_I\tilde{L})=\rho(\hat{D}+d_I(\hat{P}-\mathbb{I}_n)\hat{D})
=s(\hat{D}+d_I(\hat{P}-\mathbb{I}_n)\hat{D})\]
satisfies $\rho'(d_I)<0$ and $\rho''(d_I)\ge0$. It follows from \eqref{dR0} and \eqref{ddR0} that $\R_0'(d_I)<0$ and $\R_0''(d_I)>0$. Therefore, the strict monotonicity and strict convexity of $\R_0(d_I)$ follows.
\end{proof}

Biologically, fast diffusion of the infected subpopulation decreases the disease transmission potential. The negativity of $\R_0'(d_I)$ and the positivity of $\R_0''(d_I)$ mean that $\R_0$ is monotone decreasing but has a positive acceleration. So the impact of increasing infected human diffusion on reducing the infection risk keeps shrinking. In particular, the fastest declining speed for $\R_0$ is achieved at $d_I=0$. Suppose that
$\R_0^{(1)}\le\R_0^{(2)}\le\cdots\le\R_0^{(n-1)}<\R_0^{(n)}$, then
\begin{equation}\label{dR0dI0}
\begin{aligned}
\R_{0}'(0)=\lim\limits_{d_I\to 0}\R_{0}'(d_I)=\dfrac{\beta_n}{\gamma_n^2}L_{nn}<0.
\end{aligned}
\end{equation}
Indeed, let $A(d_I)=k\mathbb{I}_n-\tilde{D}+d_I\tilde{L}$; when $d_I=0$, the right and left eigenvectors corresponding to the largest eigenvalue $k-1/\R_0(0)$ of matrix $A(0)=k\mathbb{I}_n-\tilde{D}$ are respectively
\[\boldsymbol{x}(0)=(0,\dots,0,1)^T\ \mbox{ and }
\ \boldsymbol{y}^T(0)=(0,\dots,0,1).\]
By our assumption, the largest eigenvalue is not repeated, then $\rho(0)=k-1/\R_0^{(n)}$ and
\[\dfrac{d\rho}{d d_I}\Big|_{d_I=0}=\boldsymbol{y}^T(0)\dfrac{dA(d_I)}{d d_I}\Big|_{d_I=0}\boldsymbol{x}(0)
=\boldsymbol{y}^T(0)\tilde{L}\boldsymbol{x}(0)=\dfrac{L_{nn}}{\beta_n}.\]
Substituting the above results into \eqref{dR0} gives \eqref{dR0dI0}.

\section{Applications}

We will demonstrate some simple applications of the approach and results obtained in previous section to the SIS epidemic patch model \eqref{sismodel}.

\subsection{Asymptotic Behavior of $\R_0$ and $s(F-V)$}

\begin{lemma}\label{lem33}
Let $L=(L_{ij})$ be an $n\times n$ matrix with zero column sum and $L^*=(L^*_{ij})^T$ be the adjoint matrix of $L$ with $L^*_{ij}$ representing the $(i,j)$ cofactor of $L$. Then
\begin{itemize}
  \item [{\rm(a)}] $L^*_{ij}=L^*_{jj}$ for $1\le i,j\le n$. In particular, if $L$ is symmetric, then $L^*_{ij}=L^*_{11}$ for $1\le i,j\le n$.
  \item [{\rm(b)}] $(L^*_{11},\dots,L^*_{nn})^T$ is either zero or an eigenvector of $L$. In addition, if $L$ is essentially nonnegative and irreducible, then $(-1)^{n-1}(L^*_{11},\dots,L^*_{nn})^T$ is strictly positive.
\end{itemize}
\end{lemma}

\begin{proof}
(a) For any $i\neq j$ and $1\le i,j\le n$, we have
\begin{equation*}
\begin{aligned}
L^*_{ij}-L^*_{jj}&=
\begin{vmatrix}
L_{11} & \cdots & L_{1j-1} & 0 & L_{1j+1} & \cdots & L_{1n} \\
\vdots & \ddots & \vdots & \vdots & \vdots & \ddots & \vdots \\
L_{i-11} & \cdots & L_{i-1j-1} & 0 & L_{i-1j+1} & \cdots & L_{i-1n} \\
L_{i1} & \cdots & L_{ij-1} & 1 & L_{ij+1} & \cdots & L_{in} \\
L_{i+11} & \cdots & L_{i+1j-1} & 0 & L_{i+1j+1} & \cdots & L_{i+1n} \\
\vdots & \ddots & \vdots & \vdots & \vdots & \ddots & \vdots \\
L_{j-11} & \cdots & L_{j-1j-1} & 0 & L_{j-1j+1} & \cdots & L_{j-1n} \\
L_{j1} & \cdots & L_{jj-1} & -1 & L_{jj+1} & \cdots & L_{jn} \\
L_{j+11} & \cdots & L_{j+1j-1} & 0 & L_{j+1j+1} & \cdots & L_{j+1n} \\
\vdots & \ddots & \vdots & \vdots & \vdots & \ddots & \vdots \\
L_{n1} & \cdots & L_{nj-1} & 0 & L_{nj+1} & \cdots & L_{nn}
\end{vmatrix}=0
\end{aligned}
\end{equation*}
due to the zero column sum of the corresponding matrix. If $L$ is symmetric, so is $L^*$. Hence $L^*_{ij}=L^*_{1j}=L^*_{j1}=L^*_{11}$ for $1\le i,j\le n$.

(b) It follows from $LL^*=(\det L)I_n=0$ that the $(i,j)$ entry of $LL^*$ satisfies
\[\sum\limits_{k\in\Omega}L_{ik}L^*_{jk}=\sum\limits_{k\in\Omega}L_{ik}L^*_{kk}=0\Rightarrow L(L^*_{11},\dots,L^*_{nn})^T=0.\]
If $L$ is essentially nonnegative and irreducible, then $L_{kk}<0$ for $k=1,\dots,n$ and
$L^*_{ii}$ is the determinant of a diagonally dominant matrix, denoted by $\tilde{L}_{ii}$. If $\tilde{L}_{ii}$ is irreducible, by Corollary 6.2.27 in Horn and Johnson \cite{horn2013matrix}, every eigenvalue of matrix $-\tilde{L}_{ii}$ has positive real part and hence $(-1)^{n-1}L^*_{ii}=(-1)^{n-1}\det \tilde{L}_{ii}=(-1)^{n-1}(-1)^{n-1}\det (-\tilde{L}_{ii})=\det (-\tilde{L}_{ii})>0$. If $\tilde{L}_{ii}$ is reducible, then $\tilde{L}_{ii}$ is similar via a permutation to a block upper triangular matrix where each diagonal block is either a single entry or an irreducibly dominant submatrix. The result is obtained by again applying Corollary 6.2.27 in Horn and Johnson \cite{horn2013matrix} to each diagonal block.
\end{proof}

\begin{lemma}
Let $D=\diag\{\gamma_1,\dots,\gamma_n\}$ be a positive diagonal matrix and $L$ be an essentially nonnegative and irreducible matrix with zero column sum. As $d_I\to\infty$, the inverse of $V=D-d_I L$ converges to a strictly positive rank-one matrix
\[V^{-1}_{\infty}:=\lim\limits_{d_I\to\infty}V^{-1}=\dfrac{1}{\sum\limits_{i\in\Omega}\gamma_iL^*_{ii}}L^*,\]
where $L^*=(L^*_{ij})^T$ is the adjoint matrix of $L$.
\end{lemma} 

\begin{proof}
Since $V$ is a strictly diagonally dominant and irreducible $M$-matrix, the inverse of $V$ exists and it is positive. Obviously,
\[V^{-1}=\dfrac{1}{\det V}V^*,\]
where $V^*=(V^*_{ij})^T$ is the adjoint matrix of $V$ with $V^*_{ij}$ representing the $(i,j)$ cofactor of $V$. The determinant of $V$ can be written as
\[\det V=a_nd_I^n+a_{n-1}d_I^{n-1}+\cdots+a_1d_I+a_0,\]
where $a_n=(-1)^n\det L=0$ and $a_{n-1}=\sum\limits_{i\in\Omega}\gamma_i(-1)^{n-1}L^*_{ii}
=(-1)^{n-1}\sum\limits_{i\in\Omega}\gamma_iL^*_{ii}>0$. The positivity of $a_{n-1}$ comes from Lemma \ref{lem33}(b). Meanwhile, the $(i,j)$ cofactor of $V$ can be written as
\[V^*_{ij}=b_{n-1}d_I^{n-1}+\cdots+b_1d_I+b_0>0,\]
where $b_{n-1}=(-1)^{n-1}L^*_{ij}=(-1)^{n-1}L^*_{jj}>0$. Thus, the $(j,i)$ entry of $V^{-1}_{\infty}$ is
\[\lim\limits_{d_I\to\infty}\dfrac{V^*_{ij}}{\det V}
=\lim\limits_{d_I\to\infty}\dfrac{b_{n-1}d_I^{n-1}+\cdots+b_1d_I+b_0}{a_{n-1}d_I^{n-1}+\cdots+a_1d_I+a_0}
=\dfrac{b_{n-1}}{a_{n-1}}=L^*_{ij}\bigg/\sum\limits_{i\in\Omega}\gamma_iL^*_{ii}.\]
The proof is complete.
\end{proof}

Next, we improve some known results on the bounds of the basic reproduction number and the spectral bound of model \eqref{sismodel}.

\begin{theorem}\label{R0bounds}
For model \eqref{sismodel} with $d_I\in(0,\infty)$, the basic reproduction number $\R_0$ satisfies
\[\min\limits_{1\le i\le n}\R_0^{(i)}<\R_0(\infty)=\sum_{i\in\Omega} \beta_iL_{ii}^*\bigg/\sum_{i\in\Omega} \gamma_iL_{ii}^*<\R_{0}(d_I)=\rho(FV^{-1})<\R_0(0)=\max\limits_{1\le i\le n}\R_0^{(i)},\]
where $\R_0^{(i)}=\beta_i/\gamma_i$ and $L^*=(L^*_{ij})^T$ is the adjoint matrix of $L$.
\end{theorem}

\begin{proof}
The result that the multipatch reproduction number $\R_0$ is between the minimum and maximum patch reproduction numbers was proved by Gao and Ruan \cite{gao2011sis}. Indeed, this can be established by multiplying both sides of \eqref{eqn1} by $\boldsymbol{1}=\{1,\dots,1\}$, i.e.,
\[\boldsymbol{1}(FD^{-1}-\R_0\mathbb{I}_n)D\boldsymbol{v}=0,\]
where $FD^{-1}-\R_0\mathbb{I}_n=\diag\{\R_0^{(1)}-\R_0,\dots,\R_0^{(n)}-\R_0\}$ and $D\boldsymbol{v}\gg\boldsymbol{0}$.
It suffices to consider
\[\R_0(\infty):=\lim\limits_{d_I\to\infty}\R_0(d_I)=\lim\limits_{d_I\to\infty}\rho(FV^{-1})
=\rho\Big(\lim\limits_{d_I\to\infty}(FV^{-1})\Big)
=\rho\Big(F\lim\limits_{d_I\to\infty}V^{-1}\Big)=\rho(FV^{-1}_{\infty}).\]
The strictly positive matrix $FV^{-1}_{\infty}$ satisfies
\[\boldsymbol{1} FV^{-1}_{\infty}=(\beta_1,\dots,\beta_n)V^{-1}_{\infty}
=\dfrac{1}{\sum\limits_{i\in\Omega}\gamma_iL^*_{ii}}(\beta_1,\dots,\beta_n)L^*
=\dfrac{\sum\limits_{i\in\Omega}\beta_iL^*_{ii}}{\sum\limits_{i\in\Omega}\gamma_iL^*_{ii}}\boldsymbol{1}.\]
The proof is complete via the Perron-Frobenius theorem and the strict monotonicity of $\R_0$ with respect to $d_I$.
\end{proof}

The distribution of infected individuals as $d_I\to\infty$ is proportional to the positive vector $(-1)^{n-1}(L^*_{11},\dots,L^*_{nn})^T$.
The larger lower bound of $\R_0$ is the ratio of the average transmission rate $\sum\nolimits_{\beta}:=(-1)^{n-1}\sum\limits_{i\in\Omega}\beta_iL^*_{ii}$ to the average recovery rate $\sum\nolimits_{\gamma}:=(-1)^{n-1}\sum\limits_{i\in\Omega}\gamma_iL^*_{ii}$. Similar to Allen et al. \cite{allen2007asymptotic}, we call a patchy environment $\Omega$ a {\it low-risk domain} if
\[\sum\nolimits_{\beta}<\sum\nolimits_{\gamma},\]
but a {\it high-risk domain} if
\[\sum\nolimits_{\beta}\ge\sum\nolimits_{\gamma}.\]

Using Theorems \ref{R0monotonicity} and \ref{R0bounds}, we can easily obtain a generalization of Theorem 1 in Allen et al. \cite{allen2007asymptotic} as follows.

\begin{corollary} For model \eqref{sismodel}, suppose that $\R_0(0)=\max\limits_{i\in\Omega}\R_0^{(i)}>1$. The following hold:
\begin{itemize}
  \item [{\rm(a)}] In a low-risk domain, there exists a unique threshold value $d_I^*\in(0,\infty)$ determined by the polynomial equation $\det(F-V)=\det(F-D+d_IL)=0$ such that $\R_0>1$ for $d_I<d_I^*$ and $\R_0<1$ for $d_I>d_I^*$.
  \item [{\rm(b)}] In a high-risk domain, we have $\R_0>1$ for all $d_I\ge 0$.
\end{itemize}
\end{corollary}

With respect to the spectral bound of $F-V$, the following is a generalization of Lemma 3.4 in Allen et al. \cite{allen2007asymptotic}.

\begin{corollary}\label{spectralbound}
The spectral bound of the Jacobian matrix of model system \eqref{sismodel} at the disease-free equilibrium, $\lambda^*:=s(F-V)$, satisfies
\begin{itemize}
  \item [{\rm(a)}] $\lambda^*$ is strictly decreasing in $d_I\in[0,\infty)$.
  \item [{\rm(b)}] $\lambda^*\to\max\limits_{i\in\Omega}(\beta_i-\gamma_i)$ as $d_I\to0$.
  \item [{\rm(c)}] $\lambda^*\to\sum\limits_{i\in\Omega}(\beta_i-\gamma_i)L^*_{ii}\Big/\sum\limits_{i\in\Omega}L^*_{ii}$ as $d_I\to\infty$.
  \item [{\rm(d)}] In a low-risk domain, if $\R_0(0)=\max\limits_{i\in\Omega}\R_0^{(i)}>1$, then there exists a unique $d_I^*\in(0,\infty)$ determined by the polynomial equation $\det(F-V)=\det(F-D+d_IL)=0$ such that $\lambda^*>0$ for $d_I<d_I^*$ and $\lambda^*<0$ for $d_I>d_I^*$.
  \item [{\rm(e)}] In a high-risk domain, we have $\lambda^*>0$ for all $d_I\ge 0$.
\end{itemize}
\end{corollary}

\begin{proof}
Note that (b) is obvious, while (d) and (e) follow immediately from (a) and (c). Let us show the remaining two parts.

(a) Choose $k$ large enough so that all the diagonal entries of
$\hat{D}:=k\mathbb{I}_n+F-D$ are positive and that $\hat{P}:=\mathbb{I}_n+L\hat{D}^{-1}$ is an irreducible stochastic matrix. Recall that $V=D-d_I L$. For any $d_I\geq 0$, applying Lemma \ref{lem23} to the matrix
\[k\mathbb{I}_n+F-V=(k\mathbb{I}_n+F-D)+d_IL=\hat{D}+d_I(\hat{P}-\mathbb{I}_n)\hat{D}
=\left((1-d_I)\mathbb{I}_n+d_I \hat{P}\right)\hat{D}\]
gives that
\[\rho(d_I):=\rho(k\mathbb{I}_n+F-V)=s(k\mathbb{I}_n+F-V)=k+s(F-V)\]
is strictly decreasing in $d_I$. It follows that $\lambda^*$ strictly decreases as $d_I$ increases.

(c) For sufficiently large $k$, there exists a real vector $\boldsymbol{x}\gg0$ satisfying $x_1+\cdots+x_n=1$ such that
\[(k\mathbb{I}_n+F-V)\boldsymbol{x}=((k\mathbb{I}_n+F-D)+d_IL)\boldsymbol{x}
=\rho(d_I)\boldsymbol{x}=(k+s(F-V))\boldsymbol{x},\]
or equivalently,
\[(((p-k)\mathbb{I}_n-F+D)-d_IL)\boldsymbol{x}=(p-k-s(F-V))\boldsymbol{x},\ \forall\ p\in\mathbb{R}.\]
Denote $\tilde{V}=((p-k)\mathbb{I}_n-F+D)-d_IL$. For sufficiently large $p$ such that $(p-k)\mathbb{I}_n-F+D$ is a positive diagonal matrix, then
\begin{equation}\label{eqn24}
\begin{aligned}
\tilde{V}^{-1}\boldsymbol{x}=\dfrac{1}{p-k-s(F-V)}\boldsymbol{x}.
\end{aligned}
\end{equation}

We can pick up a sequence $\{d_n\}$ satisfying $0<d_1<\cdots<d_n<\cdots$ and $\lim\limits_{n\to \infty}d_n=\infty$ such that $\boldsymbol{x}(\infty):=\lim\limits_{n\to\infty}\boldsymbol{x}(d_n)$ exists. By taking $n\to\infty$, the equation \eqref{eqn24} gives
\[\tilde{V}^{-1}_{\infty}\boldsymbol{x}(\infty)=\dfrac{1}{p-k-s_{\infty}}\boldsymbol{x}(\infty),\]
which implies $s_{\infty}:=\lim\limits_{n\to\infty}s(F-V)$ exists. It follows that
\[\boldsymbol{1}\tilde{V}^{-1}_{\infty}\boldsymbol{x}(\infty)=\dfrac{1}{p-k-s_{\infty}}\boldsymbol{1}\boldsymbol{x}(\infty),\]
that is,
\[ \dfrac{\sum\limits_{i\in\Omega}L_{ii}^*}{\sum\limits_{i\in\Omega}(p-k-\beta_i+\gamma_i)L^*_{ii}}
\sum_{i\in\Omega}x_i(\infty)=\dfrac{1}{p-k-s_{\infty}}\sum_{i\in\Omega}x_i(\infty).\]
The proof is complete by solving $s_{\infty}$.
\end{proof}

\subsection{Limiting Endemic Equilibrium}

When $\R_0>1$, the model \eqref{sismodel} has at least one endemic equilibrium, denoted by
$$E^*:=(\boldsymbol{S}^*,\boldsymbol{I}^*)=(S_1^*,\dots,S_n^*,I_1^*,\dots,I_n^*),$$
which is a positive solution to
\begin{subequations}\label{EEeqn}
\begin{align}
& d_S\sum\limits_{j\in\Omega} L_{ij} S_j^*-\beta_i\dfrac{S_i^*I_i^*}{S_i^*+I_i^*}+\gamma_i I_i^*=0,\ i\in\Omega, \label{EEeqna}\\
& d_I\sum\limits_{j\in\Omega} L_{ij} I_j^*+\beta_i\dfrac{S_i^*I_i^*}{S_i^*+I_i^*}-\gamma_i I_i^*=0,\ i\in\Omega. \label{EEeqnb}
\end{align}
\end{subequations}
Previously, Allen et al. \cite{allen2007asymptotic} and Li and Peng \cite{li2019dynamics} studied the asymptotic behavior of the endemic equilibrium as $d_S\to 0$ and $d_I\to 0$, respectively. We will study the case of $d_I\to\infty$. Allen et al. \cite{allen2008asymptotic} and Peng \cite{peng2009asymptotic} considered similar problems for an SIS reaction-diffusion model.

\begin{theorem}
For model \eqref{sismodel}, assume $\R_0(\infty)=\lim\limits_{d_I\to \infty}\R_0(d_I)=\rho(FV_{\infty}^{-1})>1$ (i.e., a high-risk domain). Then the endemic equilibrium of model \eqref{sismodel} satisfies
\[E^*\to m(\hat{S}_1,\dots,\hat{S}_n,|L^*_{11}|,\dots,|L^*_{nn}|)\gg \boldsymbol{0},\ \mbox{ as }\ d_I\to\infty,\]
where $(\hat{S}_1,\dots,\hat{S}_n)$ is the unique positive solution of
\begin{equation*}
\begin{aligned}
d_S\sum\limits_{j=1}^n L_{ij}\hat{S}_j-\beta_i\frac{|L^*_{ii}|}{\hat{S}_i+|L^*_{ii}|}\hat{S}_i
+\gamma_i|L^*_{ii}|=0,\ i\in\Omega
\end{aligned}
\end{equation*}
and
\[m=\frac{\sum_{i=1}^n (S_i(0)+I_i(0))}{\sum_{i=1}^n \hat{S}_i+\sum_{i=1}^n |L_{ii}^*|}.\]
\end{theorem}

\begin{proof}
It is clear that each entry of the endemic equilibrium $E^*$
is bounded for any $d_I>0$. So, we have (up to a sequence of $d_I$)
$$E^*\to \tilde{E}:=(\tilde{\boldsymbol{S}},\tilde{\boldsymbol{I}})
=(\tilde{S}_1,\dots,\tilde{S}_n,\tilde{I}_1,\dots,\tilde{I}_n)\ge\boldsymbol{0}\ \mbox{ as } d_I\to\infty.$$
Following equation \eqref{EEeqnb} and the irreducibility of $L$, we know either $\tilde{\boldsymbol{I}}=\boldsymbol{0}$ or $\tilde{\boldsymbol{I}}\gg\boldsymbol{0}$.

Suppose $\tilde{\boldsymbol{I}}=\boldsymbol{0}$, then the equation \eqref{EEeqna} indicates that $\tilde{\boldsymbol{S}}=\boldsymbol{S}^0$ and hence $E^*\to\tilde{E}=E_0$ as $d_I\to\infty$. It follows from $\R_0(\infty)=\lim\limits_{d_I\to \infty}\R_0(d_I)>1$ and Corollary \ref{spectralbound} that $\lambda^*(\infty)=\lim\limits_{d_I\to\infty}\lambda^*(d_I)>0$. By choosing $\varepsilon\in(0,\lambda^*(\infty))$, there is a $\tilde{d}_I>0$ so that
$$\beta_i(1-S^*_i/(S^*_i+I^*_i))<\varepsilon,\ i\in\Omega,$$
for $d_I>\tilde{d}_I$. Denote $F^*=\diag\{\beta_1S^*_1/(S^*_1+I^*_1),\dots,\beta_nS^*_n/(S^*_n+I^*_i)\}$. The equation \eqref{EEeqnb} can be rewritten in a matrix form
\[(F^*-V)(\boldsymbol{I}^*)^T=\boldsymbol{0},\]
which implies \[s(F^*-V)=0.\]
On the other hand, for $d_I>\tilde{d}_I$, it follows from
\begin{equation*}
\begin{aligned}
&F^*-V>\diag\{\beta_1-\varepsilon,\dots,\beta_n-\varepsilon\}-V=F-V-\varepsilon\mathbb{I}_n
\end{aligned}
\end{equation*}
that
\begin{equation*}
\begin{aligned}
&s(F^*-V)>&s(F-V)-\varepsilon=\lambda^*(d_I)-\varepsilon>\lambda^*(\infty)-\varepsilon>0,
\end{aligned}
\end{equation*}
which results in a contradiction. This means that $\tilde{\boldsymbol{I}}\gg\boldsymbol{0}$.

The boundedness of
\[\beta_i\frac{\tilde{I}_i}{\tilde{S}_i+\tilde{I}_i}\tilde{S}_i
-\gamma_i\tilde{I}_i,\ i\in\Omega\]
implies
\begin{equation*}
\begin{aligned}
\sum\limits_{j=1}^n L_{ij}\tilde{I}_j=0,\ i\in\Omega.
\end{aligned}
\end{equation*}
Hence, the limiting endemic equilibrium $\tilde{E}$ is a solution of the system of $2n$ equations
\begin{subequations}\label{tlimeqn}
\begin{align}
& d_S\sum\limits_{j=1}^n L_{ij}\tilde{S}_j-\beta_i\frac{\tilde{I}_i}{\tilde{S}_i+\tilde{I}_i}\tilde{S}_i
+\gamma_i\tilde{I}_i=0,\ 1\le i\le n,\label{tlimeqna}  \\
& \sum\limits_{j=1}^n L_{ij}\tilde{I}_j=0,\ 1\le i\le n-1, \label{tlimeqnb}\\
& \sum\limits_{i=1}^n (\tilde{S}_i+\tilde{I}_i)=\sum\limits_{i=1}^n (S_i(0)+I_i(0)). \label{tlimeqnc}
\end{align}
\end{subequations}
Solving \eqref{tlimeqnb}
gives
\[(\tilde{I}_1,\dots,\tilde{I}_n)=m(-1)^{n-1}(L^*_{11},\dots,L^*_{nn})
=m(|L^*_{11}|,\dots,|L^*_{nn}|),\ m>0\]
and substituting it into \eqref{tlimeqna} and \eqref{tlimeqnc} yields
\begin{equation}\label{mEEeqn}
\begin{aligned}
d_S\sum\limits_{j=1}^n L_{ij}\tilde{S}_j-\beta_i\frac{m|L^*_{ii}|}{\tilde{S}_i+m|L^*_{ii}|}\tilde{S}_i
+\gamma_im|L^*_{ii}|=0,\ i\in\Omega,
\end{aligned}
\end{equation}
and
\[m=\frac{\sum_{i=1}^n (S_i(0)+I_i(0))-\sum_{i=1}^n \tilde{S}_i}{\sum_{i=1}^n |L_{ii}^*|},\]
respectively. Denote $\hat{S}_i=\tilde{S}_i/m$ for $i\in\Omega$. The equation \eqref{mEEeqn} can be rewritten as
\begin{equation}\label{hEEeqn}
\begin{aligned}
d_S\sum\limits_{j=1}^n L_{ij}\hat{S}_j-\beta_i\frac{|L^*_{ii}|}{\hat{S}_i+|L^*_{ii}|}\hat{S}_i
+\gamma_i|L^*_{ii}|=0,\ i\in\Omega.
\end{aligned}
\end{equation}
Consider the following auxiliary system
\begin{equation}\label{auxEEeqn}
\begin{aligned}
\frac{d\hat{S}_i}{dt}=d_S\sum\limits_{j=1}^n L_{ij}\hat{S}_j-\beta_i\frac{|L^*_{ii}|}{\hat{S}_i+|L^*_{ii}|}\hat{S}_i
+\gamma_i|L^*_{ii}|,\ i\in\Omega,
\end{aligned}
\end{equation}
which is dissipative, cooperative and irreducible in $\mathbb{R}^n_+$. Let $\hat{\boldsymbol{f}}$ denote the vector field described by \eqref{auxEEeqn}. Following $\hat{\boldsymbol{f}}(\boldsymbol{0})\gg\boldsymbol{0}$ and Theorem 3.2.1 in Smith \cite{smith1995monotone}, the solution starting at the origin converges to a positive equilibrium $\omega(\boldsymbol{0})$. It is easy to check that every positive equilibrium of system \eqref{auxEEeqn} is locally asymptotically stable by computing the associated Jacobian matrix. By the theory of connecting orbits \cite{hess1991periodic}, the system \eqref{auxEEeqn} cannot have more than one positive equilibrium. Furthermore, Theorem C in Jiang \cite{jiang1994global} implies that the unique positive equilibrium $\omega(\boldsymbol{0})$ is globally asymptotically stable.

Once the equation \eqref{hEEeqn} is solved, we can then obtain
\[\tilde{S}_i=m\hat{S}_i\ \mbox{ and }\ \tilde{I}_i=m|L^*_{ii}|,\]
where
\[m=\frac{\sum_{i=1}^n (S_i(0)+I_i(0))}{\sum_{i=1}^n \hat{S}_i+\sum_{i=1}^n |L_{ii}^*|}.\]
The existence and uniqueness of the positive solution of \eqref{hEEeqn} implies the convergence of $E^*$ as $d_I\to\infty$.
\end{proof}

An easy way to calculate $\sum\limits_{i\in\Omega}\beta_i L_{ii}^*$, $\sum\limits_{i\in\Omega}\gamma_i L_{ii}^*$ and $\sum\limits_{i\in\Omega}L_{ii}^*$ is through the Laplace expansion
\begin{equation*}
\begin{aligned}
\begin{vmatrix}
x_1 & x_2 & \cdots & x_n \\
L_{21} & L_{22}  & \cdots & L_{2n} \\
\vdots & \vdots &  \ddots & \vdots \\
L_{n1} & L_{n2}  & \cdots & L_{nn} \\
\end{vmatrix}=\sum\limits_{i\in\Omega}x_iL^*_{ii}.
\end{aligned}
\end{equation*}
The above-mentioned analysis can be adopted to some other epidemic patch models in studying the monotonicity, convexity and asymptotic properties of the basic reproduction number and spectral bound which serve as threshold quantities between disease persistence and extinction \cite{auger2008ross,gao2011sis,li2009global,tien2015disease}.

\section{Discussion}

It is clear that for SIS epidemic reaction-diffusion models the basic reproduction number is a monotone decreasing function of the diffusion coefficient for the infected population (e.g., Allen et al. \cite{allen2008asymptotic}, Deng and Wu \cite{deng2016dynamics}, Li et al. \cite{li2017varying}). However, the dependence of $\R_0$ on $d_I$ for SIS epidemic patch models was generally unknown \cite{allen2007asymptotic,gao2019travel}. In this paper, by applying some recent advances in the spectral theory of linear operators \cite{altenberg2009evolutionary,altenberg2012resolvent}, we show that $\R_0$ for the SIS epidemic patch model remains strictly decreasing in $d_I$ regardless of the symmetry of the connectivity matrix. Moreover, the first and second derivatives of $\R_0$ with respect to $d_I$ are strictly negative and strictly positive for all $d_I>0$, respectively. Based on the approach and results, an improved and reachable lower bound of $\R_0$, a generalized monotone result on the spectral bound of $F-V$ and the limiting endemic equilibrium as $d_I\to\infty$ are obtained. 

The present work are applicable to epidemic patch models in which exactly one class of infected compartments migrate between patches and one transmission route is involved. In other words, the next generation matrix can be written in the form of $FV^{-1}=F(D-d_IL)^{-1}$ where $F$ and $D$ are positive diagonal matrices and $L$ is an essentially nonnegative irreducible matrix with zero column sum. For example, it works for an SIS patch model with bilinear incidence \cite{wang2004epidemic}, the SIS patch model with media effect in Gao and Ruan \cite{gao2011sis}, SIR or SIRS patch model \cite{li2009global}, SEIRS patch model in the absence of diffusion for infectious subpopulation \cite{salmani2006model}, the multipatch cholera model studied by Tien et al. \cite{tien2015disease}, and a Ross-Macdonald type malaria model with human movement analyzed by Auger et al. \cite{auger2008ross} and Cosner et al. \cite{cosner2009effects}. These suggest that diffusion can help accelerate the elimination of infectious diseases.
 
The asymmetric movement in patch models can be viewed as advection-diffusion, so it is not surprising that the basic reproduction number of the SIS model of reaction-diffusion-advection type considered by Cui and Lou \cite{cui2016spatial} is also monotone decreasing in the diffusion coefficient for the infected population $d_I$ if the advection rate is proportional to $d_I$. It is worth mentioning that based on a cholera model Tien et al. \cite{tien2015disease} derived the limit of $\R_0(d_I)$ as $d_I\to\infty$ and found that the difference of $\R_0(d_I)$ and its limit is an infinitesimal of the same order as $1/d_I$ through a Laurent series expansion. The strict monotonicity of $\R_0$ with respect to $d_I$ may fail when the SIS patch model \eqref{sismodel} is extended to a multigroup-multipatch model (Example 4.3 in Gao \cite{gao2019travel}), an SEIRS reaction-diffusion model \cite{song2019spatial}, an SIS reaction-diffusion periodic model (Theorem 2.5e in Peng and Zhao \cite{peng2012reaction}), a periodic patch model (it is easy to find a counterexample by using the constructive method in Peng and Zhao \cite{peng2012reaction}), or a reaction-diffusion model with advection (Theorem 1.4 in Cui and Lou \cite{cui2016spatial}). The influence of diffusion on disease persistence is strongly affected by model structures and model formulations and further investigations are required.

\section*{Acknowledgements}

This study was partially supported by NSFC (11601336, 11571097), Program for Professor of Special Appointment (Eastern Scholar) at Shanghai Institutions of Higher Learning (TP2015050), and Shanghai Gaofeng Project for University Academic Development Program. We sincerely thank Drs. Lee Altenberg, Jifa Jiang, Yuan Lou and Gilbert Strang for their valuable discussions and comments.

\end{document}